\documentclass[letterpaper, 10 pt, conference]{ieeeconf}  
\IEEEoverridecommandlockouts  
\usepackage{amsmath}
\usepackage{amssymb}
\usepackage{lipsum}
\usepackage[caption = false]{subfig}
\usepackage{graphicx}
\usepackage{bbm}
\usepackage{xcolor}
\usepackage[open]{bookmark}
\usepackage[utf8]{inputenc}
\usepackage{comment}
\usepackage[ruled]{algorithm2e}
\SetKwInOut{Input}{Input}
\SetKwInOut{Output}{Output}
\DontPrintSemicolon
\usepackage{hyperref}

\newtheorem{theorem}{Theorem}
\newtheorem{lemma}{Lemma}

\newenvironment{myproof}[2]{\:\:\:\:\:\textit{Proof of {#1} {#2}: \noindent}}{\hfill$\blacksquare$}
\newtheorem{proposition}{Proposition}

\DeclareMathOperator{\Hinf}{\mathcal{H}_{\infty}}

\DeclareMathOperator{\H2}{\mathcal{H}_{2}}
\title{\LARGE \bf Model Reduction of Consensus Network Systems via Selection of Optimal Edge Weights and Nodal Time-Scales}
\author{Ralph Sabbagh and Dany Abou Jaoude
\thanks{This paper is accepted for publication in the 2022 American Control Conference (ACC). © 2022 IEEE. Personal use of this material is permitted. Permission from IEEE must be obtained for all other uses, in any current or future media, including reprinting/republishing this material for advertising or promotional purposes, creating new collective works, for resale or redistribution to servers or lists, or reuse of any copyrighted component of this work in other works. This work is supported by the University Research Board (URB) at the American University of Beirut (AUB).}
\thanks{The authors are with the Control and Optimization Lab, Department of Mechanical Engineering, Maroun Semaan Faculty of Engineering and Architecture, American University of Beirut, Lebanon.}
\thanks{Corresponding author: D. Abou Jaoude, {\tt\small da107@aub.edu.lb}}}
\begin{document}
\maketitle
\thispagestyle{empty}
\pagestyle{empty}
\begin{abstract}
This paper proposes model reduction approaches for consensus network systems based on a given clustering of the underlying graph. Namely, given a consensus network system of time-scaled agents evolving over a weighted undirected graph and a graph clustering, a parameterized reduced consensus network system is constructed with its edge weights and nodal time-scales as the parameters to be optimized. $\mathcal{H}_{\infty}$- and $\mathcal{H}_2$-based optimization approaches are proposed to select the reduced network parameters such that the corresponding approximation errors, i.e., the $\Hinf$- and $\H2$-norms of the error system, are minimized. The effectiveness of the proposed model reduction methods is illustrated via a numerical example.
\end{abstract}
\section{INTRODUCTION}
High-dimensional models of dynamical network systems are increasingly studied in the context of multi-agent systems across various applications ranging from network optimization \cite{9296294} to synchronization and distributed coordination \cite{review,ZlizDuan}. In \cite{Farhat_2021,9115210}, $\Hinf$- and $\H2$-based performance analyses and network designs are performed for consensus network systems using the edge variant representation of the consensus protocol. As network systems grow in size and complexity, it becomes desirable to find lower-order approximant network models that exhibit similar input-output behavior. Structure-preserving model
reduction techniques for interconnected systems based on the
balanced truncation method are presented in \cite{CHENG20172451,AbouJaoudeBalancedTruncation,AbouJaoudeLPV}. In \cite{CHENG20172451}, a generalized balanced truncation approach is proposed in which the resulting reduced order model is amenable to a network interpretation. However, no clear relation between the original and reduced network structures is established. In fact, the reconstructed reduced network system is found to evolve exclusively over a complete graph structure. In \cite{AbouJaoudeBalancedTruncation,AbouJaoudeLPV}, spatial states are used to represent the interconnections between the agents of the network system. The balanced truncation methods proposed therein apply for systems that are strongly stable, i.e., systems that possess appropriately structured generalized gramians. In \cite{AbouJaoudeCoprimeFactors} and \cite{AbouJaoudeCoprimeFactorsIJC}, coprime factors reduction methods are proposed as extensions to the balanced truncation methods in \cite{AbouJaoudeBalancedTruncation} and \cite{AbouJaoudeLPV} for systems that are not strongly stable but are strongly stabilizable and strongly detectable. The method in \cite{AbouJaoudeCoprimeFactors} is illustrated on a consensus network system example.

In the present paper, we propose model reduction methods for finding a reduced consensus network system that approximates the input-output behavior of a given consensus network system, based on a given clustering of the underlying graph. In clustering-based model reduction techniques, major emphasis is placed on finding an optimal clustering, after which the reduced network system is constructed via projection using the corresponding clustering characteristic matrix. A survey of clustering-based model reduction techniques is found in \cite{doi:10.1146/annurev-control-061820-083817}. Instead of focusing on cluster selection, the works of \cite{8795611,cheng2020reduced} propose an $\mathcal{H}_2$-based model reduction technique for consensus network systems that optimizes the edge weights of the reduced network model for a given \textit{predefined} graph clustering. Namely, the parameterized edge weights are tuned iteratively via an optimization algorithm to minimize the $\mathcal{H}_2$-norm of the error system. The proposed algorithm was shown therein to be effective, i.e., to improve the reduction error, when the edge weights are initialized as the outcome of standard clustering-based projection methods.

In this paper, we extend the scope of the work in  \cite{8795611,cheng2020reduced} in two ways. First, we consider reduced order systems that are parameterized using \textit{both edge weights and nodal time-scales}. That is, in comparison with the works of \cite{8795611} and \cite{cheng2020reduced}, we allow for additional degrees of freedom in the reduced order system, thereby allowing to further reduce the reduction error. Second, in addition to extending the $\H2$-based model reduction method in \cite{8795611} along with the optimization approach in \cite{cheng2020reduced} to further account for optimizing the nodal time-scales, we propose a novel $\Hinf$-based model reduction method for the selection of both edge weights and nodal time-scales. Our proposed $\Hinf$-based model reduction method is analogous, yet complementary, to the methods in \cite{8795611} and \cite{cheng2020reduced} in that it optimizes a different norm of the error system. Our $\Hinf$-based model reduction method remains novel even when the edge weights are the only parameters considered for the reduced order system. To conclude, we note that the model reduction approach in \cite{cheng2020reduced} allows for dealing with directed graphs, whereas this paper and the work in \cite{8795611} focus on undirected graphs. The main contributions of our paper are summarized as follows:
\begin{itemize}
\item We propose an $\mathcal{H}_\infty$-based model reduction method for consensus network systems that finds an approximant reduced network system via iterative tuning of its parameterized edge weights and nodal time-scales.
\item We extend the $\mathcal{H}_2$-based model reduction method of \cite{8795611} by parameterizing the nodal time-scales in the reduced order system in addition to the edge weights.
\item We show the efficacy of our methods via an example, wherein our algorithms are initialized using the outcome of the Petrov-Galerkin Projection (PGP) paradigm \cite{S3}.
\end{itemize}
The addition of extra parameters, i.e., the nodal time-scales, in the reduced order system increases the complexity of the optimization problems to be solved. To solve the formulated problems, we follow the approach adopted in \cite{cheng2020reduced}, which is based on the work of \cite{convexconcave} that proposes an iterative method for handling bilinear matrix inequality constraints whose left-hand-side admits a convex-concave decomposition. Specifically, our proposed approaches optimize the edge weights for fixed nodal time-scales and the nodal time-scales for fixed edge weights. While the weights optimization problems are readily expressed in a form that can be handled by the method of \cite{convexconcave}, as will be seen, the time-scales  optimization problems require more involved manipulations to be expressed in such a form.

The paper is structured as follows. In Section \ref{sectiontwo}, the notation, preliminaries, and problem setup of the paper are presented. In Section \ref{sectionthree}, the $\mathcal{H}_{\infty}$- and $\mathcal{H}_2$-based model reduction problems are formulated; and algorithms are proposed to solve them. The effectiveness of the proposed methods is demonstrated in Section \ref{sectionfour} via a numerical example. The paper concludes with Section \ref{sectionfive}.
\section{PRELIMINARIES AND PROBLEM SETUP}\label{sectiontwo}
This section presents the notation, preliminaries, and problem setup of the paper.
\subsection{Notation}\label{twoa}
The cardinality of a set $\mathcal{M}$ is denoted by $|\mathcal{M}|$. $\mathbb{R}^n$ and $\mathbb{R}^n_{++}$ denote the sets of $n$-dimensional vectors with real and positive real entries, respectively. We drop the superscripts for $n = 1$. $\mathbb{R}^{n\times m}$ denotes the set of $n\times m$ matrices with entries in $\mathbb{R}$. $\mathbb{S}^{n}$, $\mathbb{S}_{+}^{n}$, and $\mathbb{S}_{++}^{n}$ denote the sets of $n\times n$ symmetric, positive semi-definite, and positive definite matrices, respectively. $\mathbb{D}_{++}^{n}$ denotes the set of diagonal matrices in $\mathbb{S}_{++}^{n}$.  $\mathbf{0}$ denotes a matrix of zeros of appropriate dimensions. $\mathbf{1}_n$ and $I_n$ denote the vector of ones of length $n$ and the identity matrix of size $n\times n$, respectively. For a given vector $v \in \mathbb{R}^{n}$, $\mathbf{diag}(v)$ denotes a diagonal matrix with the elements of $v$ on its diagonal. For a given matrix $V \in \mathbb{R}^{n\times m}$, $V^T$ and $\mathbf{ker}(V)$ denote the transpose and nullspace of $V$, respectively. If $n>m$ and $V$ is full column rank, then $V^L$ denotes a left inverse of $V$, i.e., $V^LV=I_m$. If $n = m$, then $\mathbf{tr}(V)$ denotes the trace of $V$. If $n=m$ and $V$ is nonsingular, then $V^{-1}$ denotes the inverse of $V$. For $n=m$, we write $V\succeq \mathbf{0}$ and $V\succ\mathbf{0}$ to mean that $V\in\mathbb{S}_{+}^{n}$ and $V\in\mathbb{S}_{++}^{n}$, respectively. Given matrices $A$ in $\mathbb{R}^{n\times n}$, $B$ in $\mathbb{R}^{n\times p}$, and $C$ in $\mathbb{R}^{q\times n}$, the triple $(A,B,C)$ denotes a state-space representation of a strictly proper, continuous-time, linear time-invariant (LTI) system $\mathbf{\Sigma}$ of order $n$. The transfer function matrix of $\mathbf{\Sigma}$ is given by $\mathbf{\Sigma}(s)=C(sI_{n}-A)^{-1}B$, where $s$ denotes the Laplace variable. The system $\mathbf{\Sigma}$ is \textit{bounded input-bounded output (BIBO)} stable if and only if the poles of $\mathbf{\Sigma}$ have negative real parts. If $\mathbf{\Sigma}$ is BIBO stable, then $\|\mathbf{\Sigma}\|_{\Hinf}$ and $\|\mathbf{\Sigma}\|_{\mathcal{H}_2}$ denote the (finite) $\Hinf$-norm and $\mathcal{H}_2$-norm of $\mathbf{\Sigma}$, respectively.
\subsection{Graph Theory}\label{twob}
$\mathcal{G}(\mathcal{V},\mathcal{E},E,W)$ denotes a simple, undirected, weighted, and connected graph $\mathcal{G}$ with a set of nodes $\mathcal{V} = \{1,2,\ldots,n\}$ and a set of edges $\mathcal{E} \subseteq \mathcal{V} \times \mathcal{V}$. The elements of $\mathcal{E}$ consist of uniquely labeled, unordered pairs $(i,j)$ that indicate the existence of an edge with label $l$ between nodes $i$ and $j$, where $i$, $j\in\mathcal{V}$ and the label $l$ is in $\{1,2,\ldots,|\mathcal{E}|\}$. The nodal time-scales matrix $E \in \mathbb{D}_{++}^n$ is defined as $E:=\mathbf{diag}((e_1,e_2,\ldots,e_n))$, where $e_i$ denotes the time-scale associated with node $i\in\mathcal{V}$. Likewise, the edge weights matrix $W \in \mathbb{D}_{++}^{|\mathcal{E}|}$ is defined as $W:=\mathbf{diag}((w_1,w_2,\ldots,w_{|\mathcal{E}|}))$, where $w_l$ denotes the edge weight associated with edge $l\in\{1,2,\ldots,|\mathcal{E}|\}$. The incidence matrix $D \in \mathbb{R}^{n\times |\mathcal{E}|}$ characterizes the incidence relation between the nodes and edges of $\mathcal{G}$. $D$ is defined such that $[D]_{il} = +1$ if edge $l$ is directed from node $i$, $-1$ if edge $l$ is directed towards node $i$, and $0$ otherwise. For undirected graphs, the incidence matrix is generated by assigning arbitrary orientations to the edges of $\mathcal{G}$.
The (unscaled) Laplacian matrix $L$ of $\mathcal{G}$ is defined as $L := DWD^T\in \mathbb{S}_+^n$ and satisfies $\mathbf{ker}(L)=\mathbf{span}\{\mathbf{1_n}\}$, $L\mathbf{1}_n=0$, and $\mathbf{1}_n^TL=0$, where $\mathbf{span}\{v\}$ denotes the set of vectors of the form $\alpha v$ with $\alpha \in \mathbb{R}$ \cite{doi:10.1146/annurev-control-061820-083817}. Also, for all $i$, $j\in\mathcal{V}$, $[L]_{ij}\leq0$ if $i\neq j$ and $[L]_{ii}>0$ otherwise.

A graph clustering of $\mathcal{G}$ is a partition of $\mathcal{V}$ into $r$ nonempty, mutually exclusive subsets or ``clusters" $\mathcal{C}_1, \mathcal{C}_2,\ldots,\mathcal{C}_r$ such that $\mathcal{C}_1 \cup \mathcal{C}_2 \cup \cdots \cup \mathcal{C}_r = \mathcal{V}$. The graph clustering can be characterized by a binary matrix $\Pi \in \mathbb{R}^{n\times r}$ such that $[\Pi]_{ij} = 1$ if $i \in \mathcal{C}_j$ and $0$ otherwise. By definition, it follows that $\Pi\mathbf{1}_r=\mathbf{1}_n$. A useful property of the clustering matrix $\Pi$ is that the matrix $\hat{L}=\hat{D}\hat{W}\hat{D}^{T}\in\mathbb{S}_+^r$, where $\hat{D}$ is constructed using $D$ and $\Pi$, can be interpreted as a Laplacian matrix that characterizes a simple, undirected, weighted, and connected \textit{reduced} graph $\hat{\mathcal{G}}(\hat{\mathcal{V}},\hat{\mathcal{E}},\hat{E},\hat{W})$ that has \textit{fewer} nodes and edges than $\mathcal{G}$ \cite{doi:10.1146/annurev-control-061820-083817}. In particular, $\hat{\mathcal{V}}$ is the reduced set of nodes such that $|\hat{V}|=r$, where $r\leq n$ by definition, and the matrices $\hat{E}\in\mathbb{D}_{++}^{r}$, $\hat{W}\in\mathbb{D}_{++}^{|\hat{\mathcal{E}}|}$, and $\hat{D}\in\mathbb{R}^{r\times|\hat{\mathcal{E}}|}$ are the reduced nodal time-scales, reduced edge weights, and reduced incidence matrices of $\hat{\mathcal{G}}$, respectively. The matrices $\hat{E}$ and $\hat{W}$ can be  assigned independently of the choice of the clustering matrix $\Pi$, which is a feature exploited in the present work for optimizing the reduced order network system. The reduced graph $\hat{\mathcal{G}}$ is obtained from $\mathcal{G}$ as follows. First, the edges between the nodes within the same cluster are removed and the nodes within each cluster are aggregated into a single node. Second, if there is at least one edge between any pairs of nodes in different clusters, then a single edge between the corresponding clusters is retained. Otherwise, no edge exists between the two clusters. $\hat{D}$ is algebraically obtained by removing all duplicate and zero columns from $\Pi^TD$.
\subsection{Problem Setup}\label{twoc}
Consider a consensus network system $\mathbf{\Sigma}$ consisting of time-scaled single integrator agents with weighted interconnections evolving over a network described by a simple, undirected, weighted, and connected graph $\mathcal{G}(\mathcal{V},\mathcal{E},E,W)$. The agents and the interconnections between them are represented by the nodes and edges of $\mathcal{G}$, respectively. The dynamics of this network system are given by
\begin{equation}
    \mathbf{\Sigma}:
    \left\{
    \begin{array}{rl}
        E\dot{x}(t)\!\!\!\!&=-Lx(t) + Fu(t),\\
        y(t)\!\!\!\!&=Hx(t),\label{originalsystem}
    \end{array}
    \right.
\end{equation}
where $x(t) \in \mathbb{R}^n$, $u(t) \in \mathbb{R}^p$, and $y(t) \in \mathbb{R}^q$ are the state, input, and  output vectors at time $t\geq0$, respectively. $L=DWD^{T}\in\mathbb{S}_{+}^{n}$ and $D\in\mathbb{R}^{n \times|\mathcal{E}|}$ are the Laplacian and incidence matrices of $\mathcal{G}$. $F \in \mathbb{R}^{n\times p}$ and $H \in \mathbb{R}^{q\times n}$ are the input and output matrices of $\mathbf{\Sigma}$, respectively.

Given a clustering of the graph $\mathcal{G}$ with corresponding characteristic clustering matrix $\Pi\in\mathbb{R}^{n\times r}$, the problem addressed in this paper is to find a reduced consensus network system $\mathbf{\hat{\Sigma}}$ that 1) consists of $r$ single integrator agents whose dynamics evolve over the network described by the reduced, simple, undirected, weighted, and connected graph $\hat{\mathcal{G}}(\hat{\mathcal{V}},\hat{\mathcal{E}},\hat{E},\hat{W})$ and 2) approximates the input-output behavior of the original network system $\mathbf{\Sigma}$. The dynamics of the reduced network system are given by
\begin{equation}\label{reducedsystem}
    \mathbf{\hat{\Sigma}}:
    \left\{
    \begin{array}{rl}
        \hat{E}\dot{\hat{x}}(t)\!\!\!\!&=-\hat{L}\hat{x}(t) + \beta\hat{F}u(t),\\
        \hat{y}(t)\!\!\!\!&=\alpha\hat{H}\hat{x}(t),
    \end{array}
    \right.
\end{equation}
where $\hat{x}(t) \in \mathbb{R}^r$ and $\hat{y}(t) \in \mathbb{R}^q$  are the state and output vectors of the reduced system, respectively. $\hat{L}=\hat{D}\hat{W}\hat{D}^{T}\in\mathbb{S}_{+}^{r}$ and $\hat{D}\in\mathbb{R}^{r \times|\hat{\mathcal{E}}|}$ are the Laplacian and incidence matrices of the reduced graph $\hat{\mathcal{G}}$. As mentioned in Section \ref{twob}, $\hat{D}$ is obtained by removing the duplicate and zero columns from $\Pi^TD$. $\hat{F}=\Pi^{T}F \in \mathbb{R}^{r\times p}$ and $\hat{H}=H\Pi \in \mathbb{R}^{q\times r}$ are the reduced system input and output matrices, respectively. The reduced order system $\mathbf{\hat{\Sigma}}$ defined in \eqref{reducedsystem} is parameterized by the reduced edge weight matrix $\hat{W}$ and the reduced nodal time-scale matrix $\hat{E}$. The $\Hinf$-norm and the $\mathcal{H}_2$-norm of the error system $\mathbf{\Sigma}-\mathbf{\hat{\Sigma}}$ will be used to quantify the reduction error. Specifically, the diagonal entries in the matrices $\hat{W}$ and $\hat{E}$ are unknown parameters that will be chosen to minimize the approximation errors ${\|\mathbf{\Sigma} - \mathbf{\hat{\Sigma}}\|}_{\mathcal{H}_\infty}$ and ${\|\mathbf{\Sigma} - \mathbf{\hat{\Sigma}}\|}_{\mathcal{H}_2}$. Finally, $\alpha$ and $\beta$ are scalars that will be chosen appropriately to ensure that the error system $\mathbf{\Sigma}-\mathbf{\hat{\Sigma}}$ is BIBO stable for all $\hat{W}\in\mathbb{D}_{++}^{|\hat{\mathcal{E}}|}$ and $\hat{E}\in\mathbb{D}_{++}^{r}$. We formalize the problems addressed in this paper as follows.

Given the network system $\mathbf{\Sigma}$ defined in \eqref{originalsystem} and the parameterized reduced network system $\hat{\mathbf{\Sigma}}$ defined in \eqref{reducedsystem} for a predetermined clustering matrix $\Pi\in\mathbb{R}^{n\times r}$, solve the following optimization problems:
\begin{equation}\label{problemone}
    \begin{array}{cc}
        \displaystyle\min_{\hat{E}\in\mathbb{D}_{++}^{r},\hat{W}\in\mathbb{D}_{++}^{|\hat{\mathcal{E}}|}} & \|\mathbf{\Sigma} - \mathbf{\hat{\Sigma}}\|_{\Hinf},
    \end{array}
\end{equation}
\begin{equation}\label{problemtwo}
    \begin{array}{cc}
        \displaystyle\min_{\hat{E}\in\mathbb{D}_{++}^{r},\hat{W}\in\mathbb{D}_{++}^{|\hat{\mathcal{E}}|}} & \|\mathbf{\Sigma} - \mathbf{\hat{\Sigma}}\|_{\mathcal{H}_2}.
    \end{array}
\end{equation}
Although the PGP paradigm is applied in most clustering-based model reduction techniques, applying this paradigm limits the flexibility of constructing a more accurate reduced model once a clustering has been chosen \cite{8795611}. Namely, in this paradigm, the reduced system matrices are directly determined via projection using the characteristic clustering matrix $\Pi$, i.e., $\hat{E}=\Pi^TE\Pi$ and $\hat{L}=\Pi^TL\Pi$. In \cite{8795611}, Problem \eqref{problemtwo} is solved for $\textit{fixed}$ $\hat{E}$, with $\hat{E}=\Pi^TE\Pi$, whereby the $\H2$-norm of the error system is minimized by optimizing over the matrix $\hat{W}$. The contribution of this paper is thus to address both Problems \eqref{problemone} and \eqref{problemtwo}, wherein both $\hat{E}$ and $\hat{W}$ are free parameter matrices to be optimized. The setup in this paper is limited to undirected graphs; and future work will look at treating both problems for the case of directed graphs, see also the related work in \cite{cheng2020reduced}.
\section{MAIN RESULTS}\label{sectionthree}
This section presents the main results of the paper. In Section \ref{threea}, an appropriate choice of the scalars $\alpha$ and $\beta$ in \eqref{reducedsystem} is found that ensures that the error system $\mathbf{\Sigma}-\mathbf{\hat{\Sigma}}$ is BIBO stable for all $\hat{W}\in\mathbb{D}_{++}^{|\hat{\mathcal{E}}|}$ and $\hat{E}\in\mathbb{D}_{++}^{r}$. In Section \ref{threeb} and Section \ref{threec}, $\Hinf$- and $\mathcal{H}_2$-based optimization algorithms are proposed to search for suboptimal edge weights and nodal time-scales for $\mathbf{\hat{\Sigma}}$ that minimize ${\|\mathbf{\Sigma} - \mathbf{\hat{\Sigma}}\|}_{\mathcal{H}_\infty}$ and ${\|\mathbf{\Sigma} - \mathbf{\hat{\Sigma}}\|}_{\mathcal{H}_2}$, respectively.
\subsection[BIBO Stability of the Error System]{BIBO Stability of $\mathbf{\Sigma} - \mathbf{\hat{\Sigma}}$}\label{threea}
Consider the error system $\mathbf{\Sigma_e}=\mathbf{\Sigma-\mathbf{\hat{\Sigma}}}$. A realization for $\mathbf{\Sigma_e}$ is given by the triple $(A_e,B_e,C_e)$, where
\begin{align*}
A_e&=\begin{bmatrix}
-E^{-1}L & \mathbf{0}\\
\mathbf{0} & -\hat{E}^{-1}\hat{L}
\end{bmatrix},~B_e=\begin{bmatrix}E^{-1}F\\\beta\hat{E}^{-1}\hat{F}\end{bmatrix},\\C_e&=\begin{bmatrix}
H & -\alpha\hat{H}
\end{bmatrix}.
\end{align*}
The corresponding $q\times p$ transfer function matrix is given by $\mathbf{\Sigma_e}(s)=C_e(sI_{(n+r)}-A_e)^{-1}B_e$. By construction, and from the properties of the weighted and scaled Laplacian $E^{-1}L$ of $\mathcal{G}$ and the reduced weighted and scaled Laplacian $\hat{E}^{-1}\hat{L}$ of $\hat{\mathcal{G}}$ derived in \cite{Farhat_2021}, the spectrum of $A_e$ contains two zero eigenvalues, with the remaining eigenvalues being negative real. The following lemma provides a condition on $\alpha$ and $\beta$ to ensure that the error system $\mathbf{\Sigma_e}$ is BIBO stable, thereby guaranteeing that the approximation errors $\|\mathbf{\Sigma_e}\|_{\Hinf}$ and $\|\mathbf{\Sigma_e}\|_{\mathcal{H}_2}$ are bounded for all $\hat{W}\in\mathbb{D}_{++}^{|\hat{\mathcal{E}}|}$ and $\hat{E}\in\mathbb{D}_{++}^r$.
\begin{lemma}\label{lemmaone}
Consider the network system $\mathbf{\Sigma}$ defined in \eqref{originalsystem}  and the reduced network system $\mathbf{\hat{\Sigma}}$ defined in \eqref{reducedsystem} for a given clustering matrix $\Pi\in\mathbb{R}^{n\times r}$.
If $ \alpha\beta = \mathbf{tr}(\hat{E})/\mathbf{tr}(E)$, then $\mathbf{\Sigma_e}=\mathbf{\Sigma}-\mathbf{\hat{\Sigma}}$ is BIBO stable for all $\hat{W}\in \mathbb{D}_{++}^{|\hat{\mathcal{E}}|}$ and $\hat{E} \in \mathbb{D}_{++}^{r}$.
\end{lemma}
\begin{proof}
A similar realization $(T^{-1}A_eT,T^{-1}B_e,C_eT)$ for the error system $\mathbf{\Sigma_e}$ is first found that isolates the zero eigenvalues of $A_e$ from the remaining ones. Then, it is shown that the poles corresponding to the zero eigenvalues of $A_e$ cancel out in the matrix transfer function of the error system when $\alpha\beta = \mathbf{tr}(\hat{E})/\mathbf{tr}(E)$. As a result, in this case, the matrix transfer function of the error system contains poles lying exclusively in the left-half of the $s$-plane for all $\hat{W}\in \mathbb{D}_{++}^{|\hat{\mathcal{E}}|}$ and $\hat{E} \in \mathbb{D}_{++}^{r} $, which ensures the BIBO stability of the error system.
The sought nonsingular similarity transformation matrix $T$ and its inverse are constructed as follows:
   \begin{align*}
        T &= \begin{bmatrix}
   \frac{1}{\mathbf{tr}(E)}\mathbf{1}_n & \mathbf{0} & E^{-1}S_n & \mathbf{0}\\
   \mathbf{0} & \frac{1}{\mathbf{tr}(\hat{E})}\mathbf{1}_r & \mathbf{0} & \hat{E}^{-1}S_r
   \end{bmatrix},\\
         T^{-1} &= \begin{bmatrix}
   \mathbf{1}_n^{T}E & \mathbf{0}\\
   \mathbf{0} & \mathbf{1}_r^{T}\hat{E}\\
   S_n^{L}E & \mathbf{0}\\
   \mathbf{0} & S_r^{L}\hat{E}
   \end{bmatrix},
   \end{align*}
   where the matrices
      \begin{align}
      S_n = \begin{bmatrix}-I_{n-1}\\\mathbf{1}_{n-1}^{T}\end{bmatrix},~
      S_n^{L} = (S_n^{T}E^{-1}S_n)^{-1}S_n^{T}E^{-1}\label{SnL},
\end{align}
satisfy $\mathbf{1}_n^TS_n = \mathbf{0}$ and $S_n^LS_n = I_{n-1}$, and the matrices $S_r$ and $S_r^L$ are defined similarly. The similarity transformation $T$ is now applied. The resulting expression for ${\mathbf{\Sigma_e}}(s)$ obtained from the realization $(T^{-1}A_eT,T^{-1}B_e,C_eT)$ is given by
\begin{align*}
        &{\mathbf{\Sigma_e}}(s)= C_eT\left(sI_{(n+r)}-T^{-1}A_eT\right)^{-1}T^{-1}B_e\\
              &= \begin{bmatrix}
   \tilde{C}_e & \bar{C}_e
\end{bmatrix}\begin{bmatrix}
                    (sI_{2})^{-1} & \mathbf{0}\\
                    \mathbf{0} & (sI_{m}-\bar{A}_e)^{-1}
                 \end{bmatrix}\begin{bmatrix}
   \tilde{B}_e \\ \bar{B}_e
\end{bmatrix}\\&= \frac{1}{s}\tilde{C}_e\tilde{B}_e + \bar{C}_e(sI_{m}-\bar{A}_e)^{-1}\bar{B}_e \\
&= \frac{1}{s}\left(\frac{1}{\mathbf{tr}(E)}H\mathbf{1}_n\mathbf{1}_n^{T}F - \frac{\alpha\beta}{\mathbf{tr}(\hat{E})}\hat{H}\mathbf{1}_r\mathbf{1}_r^{T}\hat{F}\right) + {\mathbf{\bar{\Sigma}_e}}(s),
\end{align*}
where ${\mathbf{\bar{\Sigma}_e}}(s)=\bar{C}_e(sI_{m}-\bar{A}_e)^{-1}\bar{B}_e$,
\begin{align*}
    \bar{A}_e&=\begin{bmatrix}
      -S_n^{L}LE^{-1}S_n  & \mathbf{0}\\
      \mathbf{0}  & -S_r^{L}\hat{L}\hat{E}^{-1}S_r\label{ae}
    \end{bmatrix},\\
      \bar{B}_e &= \begin{bmatrix}
          S_n^{L}F \\
\beta S_r^{L}\hat{F}
      \end{bmatrix},~\bar{C}_e = \begin{bmatrix}
          HE^{-1}S_n & -\alpha\hat{H}\hat{E}^{-1}S_r
      \end{bmatrix},\\
        \tilde{B}_e &= \begin{bmatrix}
         \mathbf{1}_n^{T}F \\
\beta\mathbf{1}_r^{T}\hat{F}
      \end{bmatrix},~\tilde{C}_e = \begin{bmatrix}
     \frac{1}{\mathbf{tr}(E)}H\mathbf{1}_n & -\frac{\alpha}{\mathbf{tr}(\hat{E})}\hat{H}\mathbf{1}_r
      \end{bmatrix},
\end{align*}
and $m = n+r-2$. Thus, the applied similarity transformation isolates the zero eigenvalues of $A_e$, which means that $\bar{A}_e$ is Hurwitz and $\mathbf{\bar{\Sigma}_{e}}$ is BIBO stable. From the properties of $\Pi$ and the definitions of $\hat{F}$ and $\hat{H}$, it follows that $H\mathbf{1}_n\mathbf{1}_n^{T}F = \hat{H}\mathbf{1}_r\mathbf{1}_r^{T}\hat{F}$, and so, $\mathbf{\Sigma_e}$ is BIBO stable for $\alpha\beta = \mathbf{tr}(\hat{E})/\mathbf{tr}(E)$, which concludes the proof.\end{proof}
In \cite{8795611}, an analogous procedure to the one followed in the proof of Lemma \ref{lemmaone} is performed for a reduced network system with \textit{fixed} time-scales, and the resulting condition is $\alpha\beta = 1$. As such, the condition therein to guarantee the BIBO stability of the corresponding error system is a special case of our result in Lemma \ref{lemmaone}. Namely, if the reduced time-scale matrix is set to $\hat{E} = \Pi^TE\Pi$ as in \cite{8795611}, we get
\begin{align*}
    \alpha\beta = \frac{\mathbf{tr}(\hat{E})}{\mathbf{tr}(E)} = \frac{\mathbf{tr}(\Pi^TE\Pi)}{\mathbf{tr}(E)} = 1.
\end{align*}
Lemma \ref{lemmaone} further guarantees that $\|\mathbf{\Sigma-\mathbf{\hat{\Sigma}}}\|_{\Hinf}$ and $\|\mathbf{\Sigma-\mathbf{\hat{\Sigma}}}\|_{\H2}$ are invariant under the specific choices of $\alpha$ and $\beta$ as long as $\alpha\beta=\mathbf{tr}(\hat{E})/\mathbf{tr}(E)$. For the remainder of the paper, we make the convenient choice of $\alpha=1$ and $\beta=\mathbf{tr}(\hat{E})/\mathbf{tr}(E)$.

\subsection[H-infinity Optimization of Error System]{$\|\mathbf{\Sigma}-\mathbf{\hat{\Sigma}}\|_{\Hinf}$ Optimization}\label{threeb}
To solve Problem \eqref{problemone}, we first obtain a characterization of $\| \mathbf{\Sigma_e}\|_{\Hinf}$ in terms of $\hat{W}$ and $\hat{E}$.
\begin{theorem}\label{theoremone}
    Consider the network system $\mathbf{\Sigma}$ defined in \eqref{originalsystem}. There exists a reduced network system $\mathbf{\hat{\Sigma}}$ defined as in \eqref{reducedsystem} such that $\|\mathbf{\Sigma_e}\|_{\Hinf} < \gamma$ if and only if there exist matrices $\hat{X} \in  \mathbb{S}_{++}^{m}$, $\hat{W} \in \mathbb{D}_{++}^{|\hat{\mathcal{E}}|}$, and $\hat{E} \in \mathbb{D}_{++}^{r}$, and a sufficiently small $\hat{\theta} \in \mathbb{R}_{++}$ such that the following inequality is satisfied:
    \begin{equation}\label{mainconstraint}
       \Phi_{\infty}(\hat{X},\hat{\gamma},\hat{W},\hat{E}) = \psi_{\infty}(\hat{X},\hat{\gamma}) + \phi_{\infty}(\hat{W},\hat{E})\prec \mathbf{0},
    \end{equation}
    where $\hat{\gamma}=\hat{\theta}\gamma$, the matrix-valued mappings $\psi_{\infty} : \mathbb{S}_{++}^{m}\times\mathbb{R}_{++} \rightarrow \mathbb{S}^{2m+p+q}$ and $\phi_{\infty} : \mathbb{D}_{++}^{|\hat{\mathcal{E}}|}\times\mathbb{D}_{++}^{r} \rightarrow \mathbb{S}^{2m+p+q}$ are given by
    \begin{align}
        \psi_{\infty}(\hat{X},\hat{\gamma}) &=\begin{bmatrix}
    \bar{A}_{en}^T\hat{X} + \hat{X}\bar{A}_{en} & \hat{X}\bar{B}_{en} & \hat{\theta}\bar{C}_{en}^T & \hat{X}J\\
    \bar{B}_{en}^T\hat{X} & -\hat{\gamma} I_p & \mathbf{0} & \mathbf{0}\\
    \hat{\theta}\bar{C}_{en} & \mathbf{0} & -\hat{\gamma} I_q & \mathbf{0}\\
    J^T\hat{X} & \mathbf{0} & \mathbf{0} & \mathbf{0}
    \end{bmatrix}\!,\nonumber\\
        \phi_{\infty}(\hat{W},\hat{E}) &=\begin{bmatrix}
    -\bar{A}_{er}^T\bar{A}_{er} & -\bar{A}_{er}^T\bar{B}_{er} & \hat{\theta}\bar{C}_{er}^T & \bar{A}_{er}^T\\
   -\bar{B}_{er}^T\bar{A}_{er} &  -\bar{B}_{er}^T\bar{B}_{er} & \mathbf{0} & \bar{B}_{er}^T\\
    \hat{\theta}\bar{C}_{er} &  \mathbf{0} & \mathbf{0} & \mathbf{0}\\
    \bar{A}_{er} & \bar{B}_{er} & \mathbf{0} & -I_{m}
    \end{bmatrix}\!,  \label{phi}
    \end{align}
    and the matrices $\bar{A}_{en}\in\mathbb{R}^{m\times m}$, $\bar{A}_{er}\in\mathbb{R}^{m\times m}$, $\bar{B}_{en}\in\mathbb{R}^{m\times p}$, $\bar{B}_{er}\in\mathbb{R}^{m\times p}$,
    $\bar{C}_{en}\in\mathbb{R}^{q\times m}$, $\bar{C}_{er}\in\mathbb{R}^{q\times m}$, and $J\in\mathbb{R}^{m\times m}$ are defined in equations \eqref{Aen} through \eqref{matrices}.
\end{theorem}
\begin{proof}
Recall the terms introduced in the proof of Lemma \ref{lemmaone}. First, we define the matrix decompositions $\bar{A}_e=\bar{A}_{en}+J\bar{A}_{er}$, $\bar{B}_e = \bar{B}_{en} + J\bar{B}_{er}$, and $\bar{C}_e =\bar{C}_{en}+\bar{C}_{er}$, where
    \begin{align}
        \bar{A}_{en} &= \begin{bmatrix}\label{Aen}
   -S_n^{L}LE^{-1}S_n & \mathbf{0}\\
   \mathbf{0} & \mathbf{0}
    \end{bmatrix},~\bar{B}_{en} = \begin{bmatrix}
S_n^{L}F \\
\mathbf{0}
\end{bmatrix},\\\bar{A}_{er} &= \begin{bmatrix}
         \mathbf{0} & -S_r^{L}\hat{L}\hat{E}^{-1}S_r  \\
         \mathbf{0} & \mathbf{0}
    \end{bmatrix},~\bar{B}_{er} = \begin{bmatrix}
\beta S_r^{L}\hat{F} \\
\mathbf{0}
\end{bmatrix},\\
        \bar{C}_{en} &= \begin{bmatrix}
             HE^{-1}S_n & \mathbf{0}
        \end{bmatrix},\\
        \bar{C}_{er} &= \begin{bmatrix}
    \mathbf{0} & -\alpha\hat{H}\hat{E}^{-1}S_r
    \end{bmatrix},~J =\begin{bmatrix}
    \mathbf{0} & \mathbf{0}\\
    I_{(r-1)} & \mathbf{0}
    \end{bmatrix}.\label{matrices}
   \end{align}
Next, from the Bounded Real Lemma \cite{BEFB:94}, it follows that $\|{\mathbf{\Sigma_e}}\|_{\Hinf} = \|{\mathbf{\bar{\Sigma}_e}}\|_{\Hinf}  < \gamma $ if and only if there exists $X \in \mathbb{S}_{++}^m$ such that \begin{equation}
\begin{bmatrix}\label{BRL}
    \bar{A_e}^TX+X\bar{A_e} & X\bar{B_e} & \bar{C_e}^T\\
     \bar{B_e}^TX & -\gamma I_p & \mathbf{0}\\
    \bar{C_e} & \mathbf{0} & -\gamma I_q
    \end{bmatrix} \prec \mathbf{0}.
\end{equation}
By the Schur complement formula and for a sufficiently large $\theta>0$, \eqref{BRL} is equivalent to the following inequality:
\begin{equation}
\begin{bmatrix}\label{BRLSCHUR}
    \bar{A_e}^TX+X\bar{A_e} & X\bar{B_e} & \bar{C_e}^T & XJ\\
     \bar{B_e}^TX & -\gamma I_p & \mathbf{0} & \mathbf{0}\\
   \bar{C_e} & \mathbf{0} & -\gamma I_q & \mathbf{0}\\
    J^TX & \mathbf{0} & \mathbf{0} & -\theta I_m
    \end{bmatrix} \prec \mathbf{0}.
    \end{equation}
Based on the decompositions of $\bar{A}_e$, $\bar{B}_e$, and $\bar{C}_e$ defined in \eqref{Aen} through \eqref{matrices}, a nonsingular transformation matrix $U$ is applied to \eqref{BRLSCHUR} to decouple the variables $X$ and $\gamma$ from the variables $\hat{W}$ and $\hat{E}$. Namely, let
\begin{equation*}
U = \begin{bmatrix}
I_{m} & \mathbf{0} & \mathbf{0} & \mathbf{0}\\
\mathbf{0} & I_p & \mathbf{0} & \mathbf{0}\\
\mathbf{0} & \mathbf{0} & I_q & \mathbf{0}\\
-\bar{A}_{er} & -\bar{B}_{er} & \mathbf{0} & I_m
\end{bmatrix}.
\end{equation*} Then, pre- and post-multiplying both sides of \eqref{BRLSCHUR} by $U^T$ and $U$, respectively,  dividing both sides by $\theta>0$, and performing the substitutions $\hat{\theta} = 1/\theta$, $\hat{X} = (1/\theta)X \succ \mathbf{0}$, and $\hat{\gamma}=(1/\theta)\gamma$, yield the equivalent matrix inequality in \eqref{mainconstraint}.
\end{proof}

Theorem \ref{theoremone} gives a characterization of the $\Hinf$-norm of the error system $\mathbf{\Sigma_e}$ that will be used to formulate an optimization problem for the selection of the edge weights and nodal time-scales of $\mathbf{\hat{\Sigma}}$. Specifically, we need to find $\hat{X}\in\mathbb{S}_{++}^m$, $\hat{W}\in\mathbb{D}_{++}^{|\hat{\mathcal{E}}|}$, and $\hat{E}\in\mathbb{D}_{++}^r$ that minimize $\hat{\gamma}$ such that \eqref{mainconstraint} holds for a given $\hat{\theta}>0$.

In \cite{cheng2020reduced}, an analogous decoupling procedure to the one used in the proof of Theorem \ref{theoremone} is applied on a matrix inequality (counterpart of (\ref{BRLSCHUR})) that characterizes the $\H2$-norm of the error system. The goal therein is to only decouple $\hat{X}$ from $\hat{W}$ (since $\hat{E}$ is fixed). In that case, the left-hand-side of the resulting inequality takes the form of the difference of two positive semidefinite-convex matrix-valued mappings, also known as a  \textit{psd-convex-concave decomposition} of the counterpart of $\Phi_{\infty}$ therein. Thus, an optimization problem for the selection of the edge weights is formulated therein that can be solved by following the procedure in \cite{convexconcave}. This procedure consists of iteratively solving a sequence of convex optimization problems and is guaranteed to converge to a locally optimal solution of the nonconvex problem of interest. Namely, starting from a feasible solution to the original problem, a convex inequality constraint is obtained by replacing the psd-concave term in the original inequality by its linearization about this initial solution. A convex optimization problem is then solved in which the original inequality constraint is replaced by its convex counterpart. In the second iteration, the psd-concave term in the original inequality is linearized about the solution of the convex optimization problem solved in the first iteration. The process repeats until convergence. As per Lemma~\ref{lemmaone}, all $\hat{W}\in\mathbb{D}_{++}^{|\hat{\mathcal{E}}|}$ are feasible starting points for this iterative algorithm. As a `good' starting point, the solution obtained from applying the PGP paradigm can be used to initialize the algorithm, which ensures that the obtained results are at least as good as the ones from this classic paradigm.

In our case, Lemma \ref{lemmaone} ensures that all $\hat{W}\in\mathbb{D}_{++}^{|\hat{\mathcal{E}}|}$ and $\hat{E}\in\mathbb{D}_{++}^{r}$ are feasible starting points for any similar iterative algorithm to be proposed. However, although $\psi_{\infty}$ is psd-convex in $(\hat{X},\hat{\gamma})$ and $\phi_{\infty}$ is psd-concave in $\hat{W}$, $\phi_{\infty}$ is \textit{not} psd-concave in $\hat{E}$. In what follows, we propose to solve Problem \eqref{problemone} by alternating between optimizing the edge weights for fixed nodal time-scales and optimizing the nodal time-scales for fixed edge weights. In particular, for \textit{fixed $\hat{E}$}, $\phi_{\infty}(\hat{W})$ \textit{is} psd-concave in $\hat{W}$, and so, to optimize the edge weights for fixed time-scales, we follow a procedure similar to the one in \cite{cheng2020reduced}. However, for \textit{fixed $\hat{W}$}, $\phi_{\infty}(\hat{E})$ \textit{is not} psd-concave in $\hat{E}$. Nonetheless, constraint \eqref{mainconstraint} in this case can still be modified and rewritten in a form amenable to solution by the procedure in \cite{convexconcave}. Namely, we treat $\hat{E}^{-1}$ as the variable instead of $\hat{E}$ and introduce an additional matrix-valued variable $Z$ such that $Z=\hat{E}$ to rewrite $\bar{B}_{er}$ as a function of $Z$ (not of $\hat{E}^{-1}$). By doing so, inequality \eqref{mainconstraint} is transformed into the following set of constraints:
\begin{align}\label{eq26}
\psi_{\infty}(\hat{X},\hat{\gamma}) + \phi_{\infty}(\hat{E}^{-1},Z) &\prec \mathbf{0},\\
\begin{bmatrix}\label{eq27}
\hat{E}^{-1} & I_{r}\\
I_{r} & Z
\end{bmatrix} &\succeq \mathbf{0},\\Z-(\hat{E}^{-1})^{-1}&\preceq \mathbf{0}\label{eq28}.
\end{align}
In other words, the added equality constraint $Z=\hat{E}$ is split into $\hat{E} \preceq Z$ and $Z\preceq \hat{E}$.
By application of the Schur complement formula, the inequality $\hat{E} \preceq Z$ is transformed into the linear matrix inequality in $(\hat{E}^{-1},Z)$ in \eqref{eq27}. The inequality $Z\preceq \hat{E}$ is written as $Z-(\hat{E}^{-1})^{-1}\preceq \mathbf{0}$ as in \eqref{eq28}, where the left-hand-side takes the form of the difference of two psd-convex matrix-valued mappings in $Z$ and $\hat{E}^{-1}$, namely, $Z$ and $(\hat{E}^{-1})^{-1}$. Thus, by introducing the additional variable $Z$ and expressing the constraint \eqref{mainconstraint} as the equivalent set of constraints \eqref{eq26}-\eqref{eq28}, we obtain the following result.
\begin{lemma}\label{lemmatwo}
For any fixed $\hat{W}\in\mathbb{D}_{++}^{|\hat{\mathcal{E}}|}$, the matrix-valued mapping $\phi_{\infty}$ in \eqref{eq26} is psd-concave in $(\hat{E}^{-1},Z)$.
\end{lemma}
To prove Lemma \ref{lemmatwo}, we first need to prove the following proposition.
\begin{proposition}\label{propone}
Consider the full and reduced order systems $\mathbf{\Sigma}$ and $\mathbf{\hat{\Sigma}}$, and assume that the nodal time-scale matrices are partitioned as in $E=\begin{bmatrix}
\bar{E} & \mathbf{0}\\
\mathbf{0} & e_n
\end{bmatrix}$ and $\hat{E}=\begin{bmatrix}
\hat{\bar{E}} & \mathbf{0}\\
\mathbf{0} & e_r
\end{bmatrix}$, where $\bar{E}\in\mathbb{D}_{++}^{n-1}$ and $\hat{\bar{E}}\in\mathbb{D}_{++}^{r-1}$, respectively. Define $S_n$, $S_r$, $S_n^L$, and $S_r^L$ as in \eqref{SnL}. Then,
\begin{align}
\mathbf{tr}(E)S_n^{L}
        &= \begin{bmatrix}
        \bar{E}\mathbf{1}_{n-1}\mathbf{1}^T_{n-1}{-}\mathbf{tr}(E)I_{n-1}&\bar{E}\mathbf{1}_{n-1}
        \end{bmatrix}\label{eqw1},\\S_n^{L}D&=-\left[I_{n-1}\:\:\:\mathbf{0}\right]D\label{eqw2},\\
        \mathbf{tr}(\hat{E})S_r^{L}
        &= \begin{bmatrix}
        \hat{\bar{E}}\mathbf{1}_{r-1}\mathbf{1}^T_{r-1}{-}\mathbf{tr}(\hat{E})I_{r-1}&\hat{\bar{E}}\mathbf{1}_{r-1}
        \end{bmatrix}\label{eqw3},\\S_r^{L}\hat{D}&=-\left[I_{r-1}\:\:\:\mathbf{0}\right]\hat{D}\label{eqw4}.
\end{align}
\end{proposition}

\begin{proof} We prove \eqref{eqw1} and \eqref{eqw2}; \eqref{eqw3} and \eqref{eqw4} are proved similarly. From the definitions of $S_n$ and $S_n^L$, we have
\begin{align*}
    S_n^L &= \left(S_n^TE^{-1}S_n\right)^{-1}S_n^TE^{-1}\\
    &= \left(\bar{E}^{-1}+e_n^{-1}\mathbf{1}_{n-1}\mathbf{1}_{n-1}^T\right)^{-1}S_n^TE^{-1}\\
    &= \left(\bar{E}^{-1}+e_n^{-1}\mathbf{1}_{n-1}\mathbf{1}_{n-1}^T\right)^{-1}\left[-\bar{E}^{-1}\:\:e_n^{-1}\mathbf{1}_{n-1}\right]\\
    &= \left[\frac{1}{\mathbf{tr}(E)}\bar{E}\mathbf{1}_{n-1}\mathbf{1}_{n-1}^T{-}I_{n-1}\:\:\:\frac{1}{\mathbf{tr}(E)}\bar{E}\mathbf{1}_{n-1}\right],
\end{align*}
where the last equality follows from the matrix inversion lemma. Multiplying both sides by $\mathbf{tr}(E)$ yields \eqref{eqw1}. Equation \eqref{eqw1} is then used to prove equation \eqref{eqw2}. Namely,
\begin{align*}
    S_n^LD&=\left(\frac{1}{\mathbf{tr}(E)}\bar{E}\mathbf{J} - \left[I_{n-1}\:\:\:\mathbf{0} \right] \right)\!D\\
&= \frac{1}{\mathbf{tr}(E)}\bar{E}(\mathbf{J}D) - \left[I_{n-1}\:\:\:\mathbf{0}\right]\!D,
\end{align*}
where $\mathbf{J}$ is the matrix of ones of size $(n-1)\times n$. The structure of the incidence matrix $D$ implies that $\mathbf{J}D=\mathbf{0}$, which completes the proof.
\end{proof}
\begin{myproof}{Lemma}{\ref{lemmatwo}}
By Proposition \ref{propone}, $\mathbf{tr}(\hat{E})S_r^{L}$ is linear in $\hat{E}$ and $S_r^L\hat{D}$ is independent of $\hat{E}$ for all $\hat{E}\in\mathbb{D}_{++}^{r}$. As such, with the assumed choice of $\alpha=1$ and $\beta=\mathbf{tr}(\hat{E})/\mathbf{tr}(E)$, $\bar{A}_{er}$ is linear in $\hat{E}^{-1}$, $\bar{B}_{er}$ is linear in $\hat{E}$, and $\bar{C}_{er}$ is linear in $\hat{E}^{-1}$. Thus, for fixed $\hat{W}$, the structure of $\phi_{\infty}$ shown in \eqref{phi} can be leveraged by introducing the additional variable $Z$ such that $Z=\hat{E}$ and expressing $\bar{B}_{er}$ in terms of $Z$ to get $\phi_{\infty}(\hat{E}^{-1},Z)$ as in \eqref{eq26}. Define the linear matrix-valued mapping $\mathcal{L}(\hat{E}^{-1},Z)=\begin{bmatrix}\bar{A}_{er} & \bar{B}_{er}\end{bmatrix}$.
Since the off-diagonal block matrices of $-\phi_{\infty}(\hat{E}^{-1},Z)$ are linear in $(\hat{E}^{-1},Z)$, it is enough to check the psd-convexity of the upper-left block of $-\phi_{\infty}(\hat{E}^{-1},Z)$ given by $(\mathcal{L}(\hat{E}^{-1},Z))^T \mathcal{L}(\hat{E}^{-1},Z)$. This term is indeed psd-convex in $(\hat{E}^{-1},Z)$ since it corresponds to the composition of the psd-convex function $X^TX$ with the linear mapping $\mathcal{L}(\hat{E}^{-1},Z)$ \cite{boyd_vandenberghe_2004}.
\end{myproof}

Based on the above discussion, two optimization schemes that tune the weights and time-scales of the reduced system independently are proposed. Depending on whether $\hat{E}$ is fixed and $\hat{W}$ is solved for or vice versa, a different formulation of the problem of minimizing the $\Hinf$-norm of the error system is solved by following the iterative procedure in \cite{convexconcave}. The iterative optimization algorithms for the $\Hinf$-tuning of edge weights for fixed nodal time-scales and the $\Hinf$-tuning of nodal time-scales for fixed edge weights are given in Algorithm \ref{A1} and Algorithm \ref{A2}, respectively. Algorithm \ref{A1} outputs an $\Hinf$-suboptimal matrix of edge weights $\hat{W}_*$ for a given matrix of time-scales $\hat{E}$ by iteratively solving the optimization problem \eqref{weightsopt} defined as follows:
\begin{equation*}
\tag{P1}
\label{weightsopt}
\begin{array}{cc}
\displaystyle\min_{\hat{X},\hat{W},\hat{\gamma}} & \hat{\gamma}\\
\textrm{s.t.}&\psi_{\infty}(\hat{X},\hat{\gamma}) + \tilde{\phi}^{(k)}_{\infty}(\hat{W})\prec\mathbf{0},
\end{array}
\end{equation*}
where $k$ denotes the iteration counter. Algorithm \ref{A2} outputs an $\Hinf$-suboptimal matrix of nodal time-scales $\hat{E}_*$ for a given matrix of edge weights $\hat{W}$ by iteratively solving the optimization problem \eqref{scalesopt} defined as follows:\begin{equation*}
\tag{P2}
\label{scalesopt}
\begin{array}{cc}
\displaystyle\min_{\hat{X},\hat{E}^{-1},Z,\hat{\gamma}} & \hat{\gamma}\\
\textrm{s.t.}&\psi_{\infty}(\hat{X},\hat{\gamma}) + \hat{\phi}^{(k)}_{\infty}(\hat{E}^{-1},Z)\prec\mathbf{0},\\
&\!\!\!\!\!\!\!\!\!\begin{bmatrix}\lambda\hat{E}^{-1} & I_{r}\\ I_{r} & Z\end{bmatrix}\succeq\mathbf{0},~Z - f^{(k)}(\hat{E}^{-1})\preceq\mathbf{0}.
\end{array}
\end{equation*}
The linearized terms $\tilde{\phi}^{(k)}_{\infty}(\hat{W})$, $\hat{\phi}^{(k)}_{\infty}(\hat{E}^{-1},Z)$, and $f^{(k)}(\hat{E}^{-1})$ in \eqref{weightsopt} and \eqref{scalesopt} at iteration $k$ are given by
\begin{align}
    &\tilde{\phi}^{(k)}_{\infty}(\hat{W})=\phi_{\infty}(\hat{W}_{(k)})+D\phi_{\infty}(\hat{W}_{(k)})[\hat{W}-\hat{W}_{(k)}]\label{derivativeofphi},\\
    &\hat{\phi}^{(k)}_{\infty}(\hat{E}^{-1},Z)=\phi_{\infty}(\hat{E}_{(k)}^{-1},Z_{(k)})\nonumber\\&+D\phi_{\infty}(\hat{E}_{(k)}^{-1},Z_{(k)})[(\hat{E}^{-1}-\hat{E}_{(k)}^{-1},Z-Z_{(k)})],\label{secondderivativeofphi}\\
    &f^{(k)}(\hat{E}^{-1})=f\left(\hat{E}^{-1}_{(k)}\right)+Df\left(\hat{E}^{-1}_{(k)}\right)[\hat{E}^{-1}-\hat{E}^{-1}_{(k)}],\label{derivativeoff}
\end{align}
and $f(X)=X^{-1}$. Let $\mathbf{\hat{w}}_{(k)}$ be a vector such that $\hat{W}_{(k)}=\mathbf{diag}(\mathbf{\hat{w}}_{(k)})$ and denote its $i$-th component by $[\mathbf{\hat{w}}_{(k)}]_i$. Then, the derivative operator in (\ref{derivativeofphi}) at iteration $k$ satisfies
\begin{align*}
D\phi_{\infty}(\hat{W}_{(k)})[\hat{W}-\hat{W}_{(k)}]&=\sum_{i=1}^{|\hat{\mathcal{E}}|} [\mathbf{\hat{w}}-\mathbf{\hat{w}}_{(k)}]_i\frac{\partial\phi_{\infty}}{\partial\mathbf{\hat{w}}_{i}}(\mathbf{\hat{w}}_{(k)}).
\end{align*}
The other derivative operators are defined similarly. In Problem \eqref{scalesopt}, $\lambda$ should be ideally set to $1$. In practice, however, $\lambda$ is set to $1\,+\,\nu$ for a small $\nu>0$ to avoid numerical problems when solving Problem \eqref{scalesopt}.
\begin{algorithm}\label{A1}
  \caption{$\Hinf$ Iterative Edge Weight Tuning}
  \Input{$E$, $\hat{E}$, $\hat{W}_{0}$, $L$, $F$, $H$, $\Pi$, $\hat{\theta}$, $\hat{\gamma}_{0}$, $i_{tot}$, $\varepsilon$}
  \Output{$\hat{W}_*$}
  Initialize $\hat{W}_{(0)}=\hat{W}_{0}$ and $\hat{\gamma}_{(0)}=\hat{\gamma}_{0}$\\
  Set iteration step $k \gets 0$\\
  \While{$k\leq i_{tot}$ \normalfont{and} $\hat{\gamma}_{(k)}-\hat{\gamma}_{(k+1)}\geq \varepsilon$}{
    Solve \eqref{weightsopt} to obtain $\hat{W}_*$, $\hat{X}_*$, and $\hat{\gamma}_*$\\
    Update $k \gets k+1$\\
    Update $\hat{W}_{(k)} \gets \hat{W}_*$ and $\hat{\gamma}_{(k)} \gets \hat{\gamma}_*$
  }
\end{algorithm}
\setlength{\intextsep}{0pt}
\begin{algorithm}\label{A2}
  \caption{$\Hinf$ Iterative Nodal Time-scale Tuning}
  \Input{$E$, $\hat{E}_{0}$, $\hat{W}$, $L$, $F$, $H$, $\Pi$, $\hat{\theta}$, $\hat{\gamma}_{0}$, $i_{tot}$, $\varepsilon$}
  \Output{$\hat{E}_*$}
  Initialize $(\hat{E}_{(0)}^{-1},Z_{(0)})=(\hat{E}_{0}^{-1},\hat{E}_{0})$ and $\hat{\gamma}_{(0)}=\hat{\gamma}_{0}$\\
  Set iteration step $k \gets 0$\\
  \While{$k\leq i_{tot}$ \normalfont{and} $\hat{\gamma}_{(k)}-\hat{\gamma}_{(k+1)}\geq \varepsilon$}{
    Solve \eqref{scalesopt} to obtain $\hat{E}_*^{-1}$, $Z_*$, $\hat{X}_*$, and $\hat{\gamma}_*$\\
    Update $k \gets k+1$\\
    Update $(\hat{E}_{(k)}^{-1},Z_{(k)})\gets(\hat{E}_*^{-1},\hat{E}_*)$ and $\hat{\gamma}_{(k)} \gets \hat{\gamma}_*$
  }
  \end{algorithm}
\setlength{\intextsep}{0pt}
\subsection[H2-based Optimization]{$\|\mathbf{\Sigma}-\mathbf{\hat{\Sigma}}\|_{\mathcal{H}_2}$ Optimization}\label{threec}
In this section, we extend the results in \cite{cheng2020reduced}, in the case of undirected graphs, to account for parameterized nodal time-scales in addition to parameterized edge weights in the reduced order system $\mathbf{\hat{\Sigma}}$.
\begin{theorem}\label{theoremthree}
Consider the network system $\mathbf{\Sigma}$ defined in \eqref{originalsystem}. There exists a reduced network system $\mathbf{\hat{\Sigma}}$ defined as in \eqref{reducedsystem} such that $\|\mathbf{\Sigma_e}\|^2_{2} < \gamma$ if and only if there exists $\hat{X} \in  \mathbb{S}_{++}^{m}$, $\hat{W} \in \mathbb{D}_{++}^{|\hat{\mathcal{E}}|}$, $\hat{E} \in \mathbb{D}_{++}^{r}$, $\hat{R} \in \mathbb{S}_{++}^q$, and a sufficiently small $\hat{\theta} \in \mathbb{R}_{++}$ such that the following inequalities are satisfied:
\begin{align}
       \Phi_2(\hat{X},\hat{W},\hat{E})=\psi_{2}(\hat{X}) + \phi_{2}(\hat{W},\hat{E})&\prec \mathbf{0}\label{h2const},\\
       \begin{bmatrix}
         \hat{X} & \hat{\theta}C_e^T\\
         \hat{\theta}C_e  & \hat{R}
       \end{bmatrix} &\succ \mathbf{0}\label{h7const},\\
       \mathbf{tr}(\hat{R})& < \hat{\gamma}\label{h11const},
    \end{align}
where $\hat{\gamma}=\hat{\theta}\gamma$, the matrix-valued mappings $\psi_{2} : \mathbb{S}_{++}^{m}\rightarrow \mathbb{S}^{2m+p}$ and $\phi_{2} : \mathbb{D}_{++}^{|\hat{\mathcal{E}}|}\times\mathbb{D}_{++}^{r} \rightarrow \mathbb{S}^{2m+p}$ are given by
    \begin{align*}
        \psi_{2}(\hat{X}) &= \begin{bmatrix}
    \bar{A}_{en}^T\hat{X} + \hat{X}\bar{A}_{en} & \hat{X}\bar{B}_{en} & \hat{X}J\\
    \bar{B}_{en}^T\hat{X} & -\hat{\theta}I_p & \mathbf{0}\\
    J^T\hat{X} & \mathbf{0} & \mathbf{0}
    \end{bmatrix},\\
        \phi_{2}(\hat{W},\hat{E}) &=  \begin{bmatrix}
    -\bar{A}_{er}^T\bar{A}_{er} & -\bar{A}_{er}^T\bar{B}_{er} & \bar{A}_{er}^T\\
    -\bar{B}_{er}^T\bar{A}_{er} &  -\bar{B}_{er}^T\bar{B}_{er} & \bar{B}_{er}^T\\
    \bar{A}_{er} & \bar{B}_{er} &  -I_{m}
    \end{bmatrix},
    \end{align*}
and the matrices $\bar{A}_{en}$, $\bar{A}_{er}$, $\bar{B}_{en}$, $\bar{B}_{er}$, $\bar{C}_{en}$, $\bar{C}_{er}$, and $J$ are defined in equations \eqref{Aen} through \eqref{matrices}.
\end{theorem}
\begin{proof}
The proof follows similarly to the proof of Theorem \ref{theoremone} and is very closely related to the proof of \cite[Theorem 1]{cheng2020reduced}. The proof uses the characterization of the $\mathcal{H}_2$-norm of a continuous-time LTI system given in \cite[Section 4.1]{Vandeberghe} and the following decoupling matrix that plays a role similar to the role of matrix $U$ in the proof of Theorem \ref{theoremone}:
\begin{equation*}
U = \begin{bmatrix}
I_{m} & \mathbf{0}  & \mathbf{0}\\
\mathbf{0} & I_p & \mathbf{0}\\
-\bar{A}_{er} & -\bar{B}_{er}  & I_m
\end{bmatrix}.\end{equation*}\end{proof}

Theorem \ref{theoremthree} gives a characterization of the $\H2$-norm of the error system $\mathbf{\Sigma_e}$ that will be used for the optimal selection of the edge weights and nodal time-scales of $\mathbf{\hat{\Sigma}}$. Namely, we want to find $\hat{X} \in  \mathbb{S}_{++}^{m}$, $\hat{W} \in \mathbb{D}_{++}^{|\hat{\mathcal{E}}|}$, $\hat{E} \in \mathbb{D}_{++}^{r}$, and $\hat{R} \in \mathbb{S}_{++}^q$ that minimize $\hat{\gamma}$ subject to (\ref{h2const}), (\ref{h7const}), and (\ref{h11const}) for a given $\hat{\theta}>0$. As before, two  schemes that tune the weights and time-scales separately are proposed. Depending on whether $\hat{E}$ or $\hat{W}$ is fixed, a different formulation of the problem of minimizing the $\H2$-norm of the error system is solved. For \textit{fixed $\hat{E}$}, the optimization problem reduces to the one in \cite{8795611} (therein $\hat{E} = \Pi^TE\Pi$ specifically), which we solve using the method of \cite{cheng2020reduced} (not the one in \cite{8795611}). However, for \textit{fixed $\hat{W}$}, $\phi_{2}(\hat{E})$ \textit{is not} psd-concave in $\hat{E}$, and so the problem cannot be directly solved by following the procedure in \cite{convexconcave}. To resolve this issue, the same techniques used in Section \ref{threeb} to deal with $\phi_{\infty}(\hat{E})$ for fixed $\hat{W}$ are applied here. The details are omitted for brevity. The iterative optimization algorithms for the $\H2$-tuning of edge weights for fixed nodal time-scales and the $\H2$-tuning of nodal time-scales for fixed edge weights are given in Algorithms \ref{A3} and \ref{A4}, respectively. Algorithm \ref{A3} returns an $\mathcal{H}_2$-suboptimal matrix of edge weights $\hat{W}_*$ for a given matrix of time-scales $\hat{E}$ by iteratively solving the optimization problem \eqref{weightsopt2} given by
\begin{equation*}
\tag{P3}
\label{weightsopt2}
\begin{array}{cc}
\displaystyle \min_{\hat{X},\hat{W},\hat{R}}& \mathbf{tr}(\hat{R})\\
\textrm{s.t.}&\!\!\!\!\!\!\!\!\!
\psi_{2}(\hat{X}) + \tilde{\phi}^{(k)}_{2}(\hat{W})\prec \mathbf{0},~\begin{bmatrix} \hat{X} & \hat{\theta}C_e^T\\ \hat{\theta}C_e  & \hat{R} \end{bmatrix} \succ \mathbf{0},
\end{array}
\end{equation*}
where the linearized term $\tilde{\phi}^{(k)}_{2}(\hat{W})$ at iteration $k$ is defined similarly to $\tilde{\phi}^{(k)}_{\infty}(\hat{W})$ in (\ref{derivativeofphi}). Algorithm \ref{A4} returns an $\mathcal{H}_2$-suboptimal matrix of nodal time-scales $\hat{E}_*$ for a given matrix of edge weights $\hat{W}$ by iteratively solving the following optimization problem \eqref{scalesopt2}, where $\hat{\phi}^{(k)}_{2}(\hat{E}^{-1},Z)$ is defined similarly to $\hat{\phi}^{(k)}_{\infty}(\hat{E}^{-1},Z)$ in \eqref{secondderivativeofphi}:
\begin{equation*}
    \tag{P4}
\label{scalesopt2}
\begin{array}{cc}
\displaystyle \min_{\hat{X},\hat{E}^{-1},Z,\hat{R}}& \mathbf{tr}(\hat{R})\\
\textrm{s.t.}
&\!\!\!\!\!\!\!\!\!\!\!\!\!\psi_{2}(\hat{X}) + \hat{\phi}^{(k)}_{2}(\hat{E}^{-1},Z) \prec \mathbf{0},~Z - f^{(k)}(\hat{E}^{-1})\preceq \mathbf{0},\\
&\!\!\!\!\!\!\!\!\!\!\!\!\!\!\begin{bmatrix} \hat{X} & \hat{\theta}C_e^T(\hat{E}^{-1})\\ \hat{\theta}C_e(\hat{E}^{-1}) &\hat{R}\end{bmatrix}\succ\mathbf{0},~\begin{bmatrix} \lambda\hat{E}^{-1} & I_{r}\\ I_{r} & Z \end{bmatrix} \succeq \mathbf{0}.
\end{array}
\end{equation*}
\begin{algorithm}\label{A3}
  \caption{$\H2$ Iterative Edge Weight Tuning}
  \Input{$E$, $\hat{E}$, $\hat{W}_0$, $L$, $F$, $H$, $\Pi$, $\hat{\theta}$, $r$, $i_{tot}$, $\varepsilon$}
  \Output{$\hat{W}_*$}
  Initialize $\hat{W}_{(0)}=\hat{W}_{0}$ and $\mathbf{tr}(\hat{R}_{(0)})=r$\\
  Set iteration step $k \gets 0$\\
  \While{$k\leq i_{tot}$ \normalfont{and} $\mathbf{tr}(\hat{R}_{(k)})-\mathbf{tr}(\hat{R}_{(k+1)})\geq\varepsilon$}{
    Solve \eqref{weightsopt2} to obtain $\hat{W}_*$, $\hat{X}_*$, and $\hat{R}_*$\\
    Update $k \gets k+1$\\
    Update $\hat{W}_{(k)}\gets \hat{W}_*$ and $\hat{R}_{(k)} \gets \hat{R}_*$
  }
\end{algorithm}
\begin{algorithm}\label{A4}
  \caption{$\H2$ Iterative Nodal Time-scale Tuning}
  \Input{$E$, $\hat{E}_0$, $\hat{W}$, $L$, $F$, $H$, $\Pi$, $\hat{\theta}$, $r$, $i_{tot}$, $\varepsilon$}
  \Output{$\hat{E}_*$}
  Initialize $(\hat{E}_{(0)}^{-1},Z_{(0)})=(\hat{E}_{0}^{-1},\hat{E}_{0})$ and $\mathbf{tr}(\hat{R}_{(0)})=r$\\
  Set iteration step $k \gets 0$\\
  \While{$k\leq i_{tot}$ \normalfont{and} $\mathbf{tr}(\hat{R}_{(k)})-\mathbf{tr}(\hat{R}_{(k+1)})\geq\varepsilon$}{
    Solve \eqref{scalesopt2} to obtain $\hat{E}_*^{-1}$, $Z_*$, $\hat{X}_*$, and $\hat{R}_*$\\
    Update $k \gets k+1$\\
    Update $(\hat{E}_{(k)}^{-1},Z_{(k)})\gets(\hat{E}_*^{-1},\hat{E}_*)$ and $\hat{R}_{(k)}\!\gets\!\hat{R}_*$}
\end{algorithm}
\section{ILLUSTRATIVE EXAMPLE}\label{sectionfour}
This section considers the example adopted in \cite{8795611} and \cite{S3} to illustrate the effectiveness of the proposed model reduction method and the edge weights and nodal time-scales tuning algorithms. The consensus network system $\mathbf{\Sigma}$ considered is characterized by a graph $\mathcal{G}(\mathcal{V},\mathcal{E},E,W)$ that has ten unscaled nodes and fifteen weighted interconnections. Figure \ref{figone} shows the underlying graphs $\mathcal{G}$ and $\hat{\mathcal{G}}$ of the original and reduced network systems $\mathbf{\Sigma}$ and $\mathbf{\hat{\Sigma}}$, respectively. A colour-coded clustering of $\mathcal{G}$ is given by $\mathcal{C}_1=\{1,2,3,4\}$, $\mathcal{C}_2=\{5,6\}$, $\mathcal{C}_3=\{7\}$, $\mathcal{C}_4=\{8\}$, and $\mathcal{C}_5=\{9,10\}$. Nodes 6 and 7 receive external inputs and are highlighted with an additional circle. Finally, the difference between the state variables of nodes 6 and 10 is measured.
The incidence matrix $D$ of $\mathbf{\Sigma}$ is given by\begin{align*}
    D =\begin{bmatrix}
  \begin{smallmatrix}
    -1 & 0 & 0 & 0 & 0 & 0 & 0 & 0 & 0 & 0 & 0 & 0 & 0 & 0 & 0\\
    0 & -1 & -1 & 0 & 0 & 0 & 0 & 0 & 0 & 0 & 0 & 0 & 0 & 0 & 0\\
    0 & 0 & 0 & -1 & -1 & -1 & 0 & 0 & 0 & 0 & 0 & 0 & 0 & 0 & 0\\
0 & 0 & 0 & +1 & 0 & 0 & -1 & 0 & 0 & 0 & 0 & 0 & 0 & 0 & 0\\
    0 & +1 & 0 & 0 & +1 & 0 & +1 & -1 & -1 & -1 & 0 & 0 & 0 & 0 & 0\\
    +1 & 0 & +1 & 0 & 0 & +1 & 0 & +1 & 0 & 0 & -1 & -1 & 0 & 0 & 0\\
   0 & 0 & 0 & 0 & 0 & 0 & 0 & 0 & +1 & 0 & +1 & 0 & -1 & -1 & -1\\
    0 & 0 & 0 & 0 & 0 & 0 & 0 & 0 & 0 & +1 & 0 & +1 & +1 & 0 & 0\\
    0 & 0 & 0 & 0 & 0 & 0 & 0 & 0 & 0 & 0 & 0 & 0 & 0 & +1 & 0\\
    0 & 0 & 0 & 0 & 0 & 0 & 0 & 0 & 0 & 0 & 0 & 0 & 0 & 0 & +1
  \end{smallmatrix}
\end{bmatrix}.
\end{align*}
The system matrices and the clustering matrix are given
\begin{align*}
     &W=\mathbf{diag}((5,\:3,\:2,\:1,\:2,\:3,\:5,\:2,\:6,\:7,\:6,\:7,\:1,\:1,\:1)),\\
     &L=\!\!\begin{bmatrix}
    \begin{smallmatrix}
    5&0&0&0&0&-5&0&0&0&0\\
    0&5&0&0&-3&-2&0&0&0&0\\
    0&0&6&-1&-2&-3&0&0&0&0\\
    0&0&-1&6&-5&0&0&0&0&0\\
    0&-3&-2&-5&25&-2&-6&-7&0&0\\
    -5&-2&-3&0&-2&25&-6&-7&0&0\\
    0&0&0&0&-6&-6&15&-1&-1&-1\\
    0&0&0&0&-7&-7&-1&15&0&0\\
    0&0&0&0&0&0&-1&0&1&0\\
    0&0&0&0&0&0&-1&0&0&1
    \end{smallmatrix}
    \end{bmatrix}\!\!,\: \Pi=\!\!\begin{bmatrix}\begin{smallmatrix}
        1 & 0 & 0 & 0 & 0\\
        1 & 0 & 0 & 0 & 0\\
        1 & 0 & 0 & 0 & 0\\
        1 & 0 & 0 & 0 & 0\\
        0 & 1 & 0 & 0 & 0\\
        0 & 1 & 0 & 0 & 0\\
        0 & 0 & 1 & 0 & 0\\
        0 & 0 & 0 & 1 & 0\\
        0 & 0 & 0 & 0 & 1\\
        0 & 0 & 0 & 0 & 1
      \end{smallmatrix}
    \end{bmatrix}\!\!,\\
    &F{=}[\begin{matrix}\begin{smallmatrix}0&0&0&0&0&1&0&0&0&0\\0&0&0&0&0&0&1&0&0&0\end{smallmatrix}\end{matrix}]^T\!\!\!\!,\:\: H{=}\!\begin{bmatrix}
      \begin{smallmatrix}
          0 & 0 & 0 & 0 & 0 & 1 & 0 & 0 & 0 & -1
         \end{smallmatrix}
    \end{bmatrix}\!,\: E{=}I_{10}.
\end{align*}
Using the characteristic clustering matrix $\Pi$, the reduced incidence matrix $\hat{D}$ of $\mathbf{\hat{\Sigma}}$ is obtained by removing the duplicate and zero columns from $\Pi^TD$:
\begin{align*}
  \hat{D}=\begin{bmatrix}\begin{smallmatrix}
       -1 & 0 & 0 & 0 & 0\\
   +1 & -1 & -1 & 0 & 0\\
   0 & +1 & 0 & -1 & -1\\
   0 & 0 & +1 & +1 & 0\\
   0 & 0 & 0 & 0 & +1
  \end{smallmatrix}
\end{bmatrix}.
\end{align*}
The remaining matrices of the parameterized reduced network system $\mathbf{\hat{\Sigma}}$ are given by
\begin{align*}
    \hat{L}&=\begin{bmatrix}
      \begin{smallmatrix}
        \hat{w}_1 & -\hat{w}_1 & 0 & 0 & 0\\
       -\hat{w}_1 & \hat{w}_1+\hat{w}_2+\hat{w}_3 & -\hat{w}_2 & -\hat{w}_3 & 0\\
        0 & -\hat{w}_2 & \hat{w}_2+\hat{w}_4+\hat{w}_5 & -\hat{w}_4 & -\hat{w}_5\\
        0 &-\hat{w}_3 & -\hat{w}_4 & \hat{w}_3+\hat{w}_4 & 0\\
        0 & 0 & -\hat{w}_5 & 0 & \hat{w}_5
      \end{smallmatrix}
    \end{bmatrix}\!\!,\:\hat{F} = \begin{bmatrix}
      \begin{smallmatrix}
      0 & 0\\
      1 & 0\\
      0 & 1\\
      0 & 0\\
      0 & 0
    \end{smallmatrix}
    \end{bmatrix},\\
    \hat{E}&=\mathbf{diag}((\hat{e}_1,\:\hat{e}_2,\:\hat{e}_3,\:\hat{e}_4,\:\hat{e}_5)),\: \hat{H}=\begin{bmatrix}
      \begin{smallmatrix}
        0 & 1 & 0 & 0 & -1
    \end{smallmatrix}
    \end{bmatrix},\\
    \hat{W}&=\mathbf{diag}((\hat{w}_1,\:\hat{w}_2,\:\hat{w}_3,\:\hat{w}_4,\:\hat{w}_5)),
\end{align*}
\begin{figure}
    \vspace{-8mm}
    \centering
    \subfloat[Original Network]{\includegraphics[scale = 0.3]{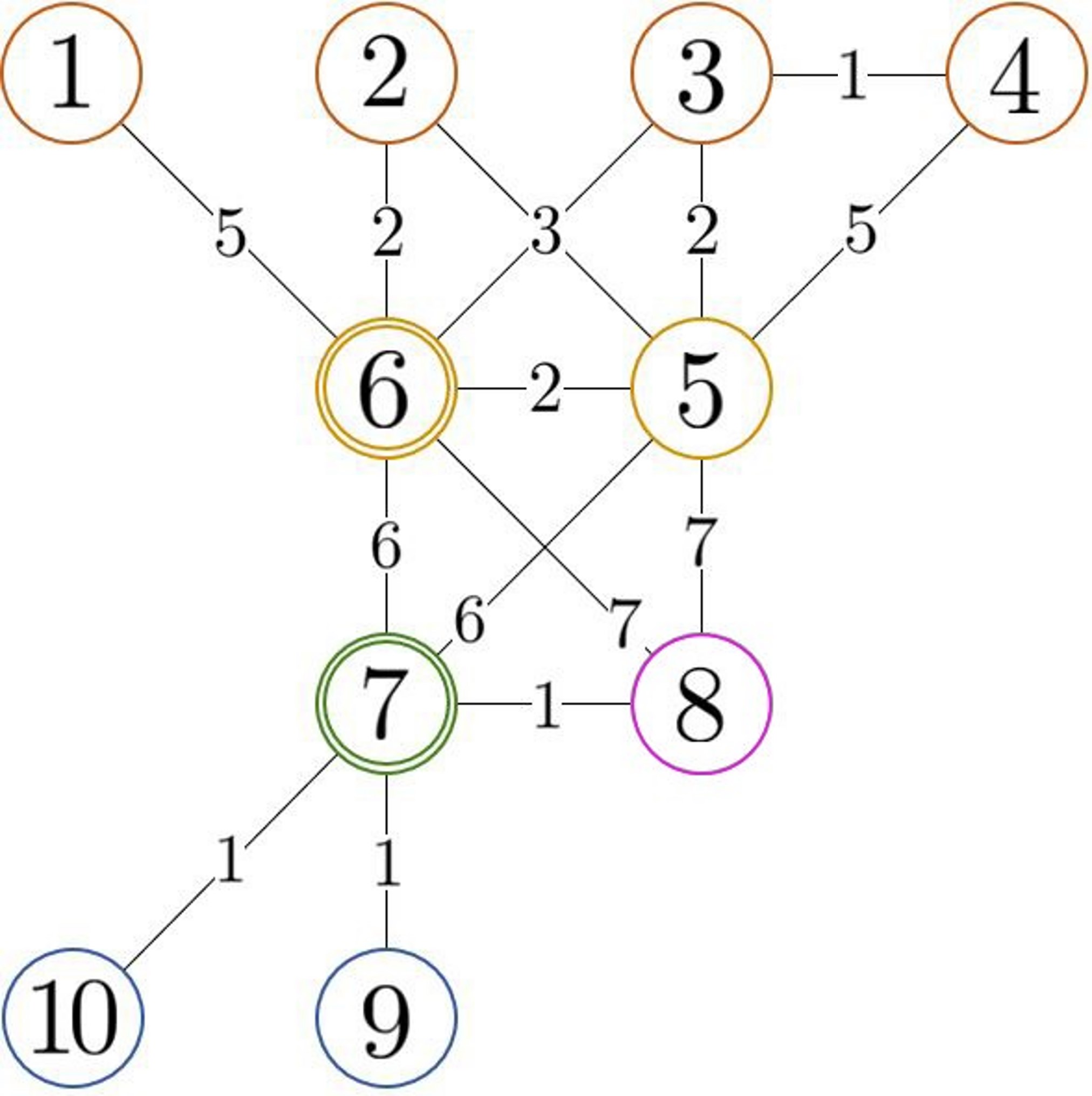}}
    \subfloat[Reduced Network]{\includegraphics[scale = 0.27]{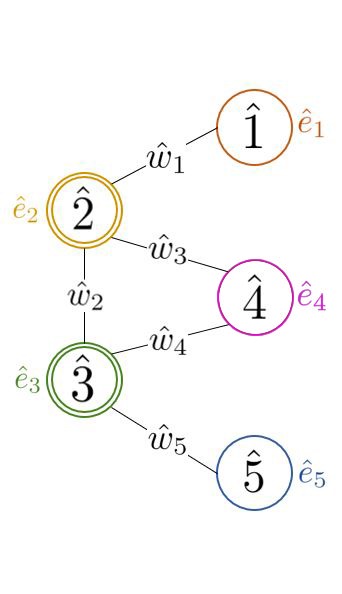}}
    \caption{(a) The underlying graph $\mathcal{G}$ of the original network system $\mathbf{\Sigma}$. (b) The underlying graph $\hat{\mathcal{G}}$ of the parameterized reduced network system $\mathbf{\hat{\Sigma}}$. The numbers inside the circles represent the node numbers, and the numbers and symbols on the connecting lines between the nodes represent the edge weights. Colour-matching nodes in the graph in (a) are clustered together to result in the graph in (b). The time-scales of the nodes in the graph in (a) are set to unity, and those of the nodes in the graph in (b) are parameterized by $\hat{e}_i$. Nodes receiving external input are highlighted with two circles.}\label{figone}
    \vspace{-4mm}
\end{figure}where $\hat{E}$ and $\hat{W}$ are the parameter matrices to be determined by the algorithms proposed in the paper.

Algorithms \ref{A1} to \ref{A4} are summarized as follows:
\begin{itemize}
\item \textbf{Algorithm \ref{A1}} (A\ref{A1}) outputs an $\Hinf$-suboptimal matrix of weights $\hat{W}_*$ for a given matrix of time-scales $\hat{E}$.
\item \textbf{Algorithm \ref{A2}} (A\ref{A2}) outputs an $\Hinf$-suboptimal matrix of time-scales $\hat{E}_*$ for a given matrix of weights $\hat{W}$.
\item \textbf{Algorithm \ref{A3}} (A\ref{A3}) outputs an $\mathcal{H}_2$-suboptimal matrix of weights $\hat{W}_*$ for a given matrix of time-scales $\hat{E}$.
\item \textbf{Algorithm \ref{A4}} (A\ref{A4}) outputs an $\mathcal{H}_2$-suboptimal matrix of time-scales $\hat{E}_*$ for a given matrix of weights $\hat{W}$.
\end{itemize}
The PGP paradigm \cite{S3} is used to determine the fixed matrix in the algorithms as well as the initial guess for the tuning. For example, $\hat{E}= \mathbf{diag}((4,\:2\:,1\:,1\:,2))$ and $\hat{W}_0= \mathbf{diag}((20,\:12\:,14\:,1\:,2))$. This is done to illustrate the advantage of the proposed optimal parameter selection algorithms in guaranteeing smaller reduction errors than the classical paradigm. Note that in this case, Algorithm \ref{A3} solves the problem considered in \cite{8795611} following the approach in \cite{cheng2020reduced}. The expression AX $\rightarrow$ AY indicates that we first perform a tuning of the edge weights using Algorithm X, and then we fix the resulting suboptimal edge weight matrix when tuning the nodal time-scales via Algorithm Y, or vice-versa. At each iteration of the algorithms, the SDPs are parsed using \texttt{YALMIP} \cite{Lofberg2004} and solved using \texttt{MOSEK} \cite{MOSEK}. We set $\hat{\theta} = 0.005$ for Algorithms \ref{A1} and \ref{A2}, $\hat{\theta} = 0.001$ for Algorithms \ref{A3} and \ref{A4}, and $\lambda = 1.0001$ for Algorithms \ref{A2} and \ref{A4}. While the algorithms are guaranteed to converge to a locally optimal point, in our implementation, the algorithms are stopped after a specified $i_{tot}$ number of iterations, and we report the achieved tolerance values $\varepsilon_{\infty}=\hat{\gamma}_{(i_{tot}-1)}-\hat{\gamma}_{(i_{tot})}$ and $\varepsilon_{2}=\mathbf{tr}(\hat{R}_{(i_{tot}-1)})-\mathbf{tr}(\hat{R}_{(i_{tot})})$ at exit.

The results are reported in Tables \ref{tableone}, \ref{tabletwo}, and \ref{tablethree}. Table \ref{tableone} and Table \ref{tabletwo} show the normalized $\mathcal{H}_2$- and $\mathcal{H}_{\infty}$-reduction errors, which are computed using the built-in \texttt{norm($\cdot\:,\:$2)} and \texttt{norm($\cdot\:,\:$inf)} functions in MATLAB, respectively. The built-in \:\texttt{minreal($\cdot$)} function is used on the original system $\mathbf{\Sigma}$ before computing $\|\mathbf{\Sigma}\|_{\Hinf}$ and $\|\mathbf{\Sigma}\|_{\H2}$. The suboptimal edge weights and nodal time-scales corresponding to the final output returned by each tuning algorithm are given in Table \ref{tablethree}, where $\hat{W}_*=\mathbf{diag}(\hat{w}_*)$ and $\hat{E}_*=\mathbf{diag}(\hat{e}_*)$.

It can be seen from the results that a significant reduction in the normalized reduction errors was achieved when the nodal time-scales are parameterized and tuned in addition to the edge weights. Compared with the PGP paradigm, Algorithms A\ref{A2} $\rightarrow$ A\ref{A1} and A\ref{A1} $\rightarrow$ A\ref{A2} reduce the normalized $\Hinf$ reduction error by $73.3\%$ and $72.6\%$, respectively. Similarly, Algorithms A\ref{A3} $\rightarrow$ A\ref{A4} and A\ref{A4} $\rightarrow$ A\ref{A3} achieve counterpart reductions in the normalized $\H2$ reduction error of $87.5\%$ and $80.6\%$, respectively. Compared to Algorithm A\ref{A3}, Algorithm A\ref{A3} $\rightarrow$ A\ref{A4} reduces the normalized $\H2$ reduction error by $75.1\%$. In the case of $\Hinf$-based model reduction, the tuning of the edge weights only, i.e., Algorithm \ref{A1}, remains a novel contribution of this work. Finally, it seems that, for this example, tuning the edge weights after tuning the nodal time-scales shows a small improvement in the normalized reduction errors compared to tuning the nodal time-scales only. For example, when applied after Algorithm \ref{A2}, Algorithm \ref{A1} reduces the normalized $\Hinf$ reduction error by $2.5\%$ only (from $0.040$ to $0.039$); and when applied after Algorithm \ref{A4}, Algorithm \ref{A3} reduces the normalized $\H2$ reduction error by $2.6\%$ only (from $0.078$ to $0.076$). It is of interest to conduct extensive simulations to validate this observation in other examples.
\begin{table}[t]
    \centering
    \caption{Normalized $\Hinf$ reduction error using Algorithm \ref{A1} and Algorithm \ref{A2}, with the corresponding values of $i_{tot}$ and $\varepsilon_{\infty}$.}
    \begin{tabular}{|c c c c|}
        \hline
         &  $\frac{\|\mathbf{\Sigma}-\mathbf{\hat{\Sigma}}\|_{\Hinf}}{\|\mathbf{\Sigma}\|_{\Hinf}}$ & $i_{tot}$ & $\varepsilon_{\infty}$ \\[1.5ex]
        \hline\hline
        PGP \cite{S3} & 0.146 & --- & ---\\[0.5ex]
        \hline
        A\ref{A1} & $0.076$& $20000$ & 1E-10\\[0.5ex]
        \hline
        A\ref{A2} & $0.040$ & $6000$ & 3E-10\\[0.5ex]
        \hline
        A\ref{A1} $\rightarrow$ A\ref{A2} & $0.040$& $20000/2000$ & 1E-10/2E-13\\[0.5ex]
        \hline
        A\ref{A2} $\rightarrow$ A\ref{A1} & $0.039$& $6000/1000$ & 3E-10/2E-10\\[0.5ex]
        \hline
    \end{tabular}
    \label{tableone}
\end{table}
\begin{table}[t]
\centering
 \caption{Normalized $\H2$ reduction error using Algorithm \ref{A3} and Algorithm \ref{A4}, with the corresponding values of $i_{tot}$ and $\varepsilon_{2}$.} \begin{tabular}{|c c c c|}
 \hline
   & $\frac{\|\mathbf{\Sigma}-\mathbf{\hat{\Sigma}}\|_{\mathcal{H}_2}}{\|\mathbf{\Sigma}\|_{\H2}}$ & $i_{tot}$ & $\varepsilon_{2}$\\[1.5ex]
 \hline\hline
PGP \cite{S3} &  0.392 & --- & ---\\[0.5ex]
 \hline
 A\ref{A3} \cite{8795611,cheng2020reduced} &  $0.313$ & $1000$ & 4E-15\\[0.5ex]
 \hline
 A\ref{A4} & $0.078$ & $2500$ & 3E-08\\[0.5ex]
 \hline
A\ref{A3} $\rightarrow$ A\ref{A4} & $0.049$ & $1000/1000$ & 4E-15/1E-12\\[0.5ex]
 \hline
A\ref{A4} $\rightarrow$ A\ref{A3} & $0.076$ & $2500/500$ & 3E-08/7E-12\\[0.5ex]
 \hline
 \end{tabular}
 \vspace{-3mm}
\label{tabletwo}
\end{table}
\begin{table}
\centering
 \begin{tabular}{|c c|}
 \hline
Algorithm & Optimized Parameters\\ [0.5ex]
 \hline\hline
A\ref{A1} & $\hat{w}_*=[12.49,\:9.50,\:0.80,\:36.44,\:1.96]$ \\
\hline
A\ref{A1} $\rightarrow$ A\ref{A2} & $\hat{w}_*=[12.49,\:9.50,\:0.80,\:36.44,\:1.96]$\\
&$\hat{e}_*=[3.74,\:1.47,\:1.61,\:1.52,\:1.91]$\\
\hline \hline
A\ref{A2} & $\hat{e}_*=[5.63,\:0.90,\:4.27,\:0.24,\:1.96]$ \\
 \hline
A\ref{A2} $\rightarrow$ A\ref{A1} & $\hat{e}_*=[5.63,\:0.90,\:4.27,\:0.24,\:1.96]$\\
&$\hat{w}_*=[19.28,\:6.48,\:31.86,\:8.30,\:2.01]$
\\
 \hline \hline
A\ref{A3} \cite{8795611,cheng2020reduced} & $\hat{w}_*=[11.30,\:10.03,\:0.04,\:11.09,\:2.10]$
\\
 \hline
A\ref{A3} $\rightarrow$ A\ref{A4} &
$\hat{w}_*= [11.30,\:10.03,\:0.04,\:11.09,\:2.10]$\\&$\hat{e}_*=[3.34,\:0.96,\:1.74,\:1.27,\:1.96]$\\
 \hline \hline
A\ref{A4} & $\hat{e}_*=[7.58,\:1.45,\:3.49,\:0.36,\:2.26]$ \\
\hline
A\ref{A4} $\rightarrow$ A\ref{A3} & $\hat{e}_*=[7.58,\:1.45,\:3.49,\:0.36,\:2.26]$\\
&$\hat{w}_*=[20.52,\:12.71,\:5.85,\:1.84,\:2.20]$  \\[1ex]
 \hline
 \end{tabular}
 \caption{Optimized edge weights and nodal time-scales returned by the proposed algorithms, rounded to two decimal places.}
\label{tablethree}
\vspace{-4mm}
\end{table}
\section{CONCLUSION}\label{sectionfive}
In this paper, $\mathcal{H}_{\infty}$- and $\mathcal{H}_2$-based model reduction methods for consensus network systems are proposed, whereby a reduced network system with parameterized edge weights and nodal time-scales is tuned via iterative optimization algorithms. Compared to the PGP paradigm in \cite{S3} and the $\H2$-based methods for optimal edge weights selection in \cite{8795611,cheng2020reduced}, our $\H2$-based methods allow for additionally optimizing the nodal time-scales of the reduced order system, thereby further reducing the reduction error. Our $\Hinf$-based approaches are novel for optimizing the edge weights, nodal time-scales, or both sequentially. Further investigation into simultaneous edge weight and nodal time-scale optimization is of interest for future work.
\bibliographystyle{IEEEtran}
\bibliography{ACC2022_References}

\begin{thebibliography}{10}
\providecommand{\url}[1]{#1}
\csname url@samestyle\endcsname
\providecommand{\newblock}{\relax}
\providecommand{\bibinfo}[2]{#2}
\providecommand{\BIBentrySTDinterwordspacing}{\spaceskip=0pt\relax}
\providecommand{\BIBentryALTinterwordstretchfactor}{4}
\providecommand{\BIBentryALTinterwordspacing}{\spaceskip=\fontdimen2\font plus
\BIBentryALTinterwordstretchfactor\fontdimen3\font minus
  \fontdimen4\font\relax}
\providecommand{\BIBforeignlanguage}[2]{{%
\expandafter\ifx\csname l@#1\endcsname\relax
\typeout{** WARNING: IEEEtran.bst: No hyphenation pattern has been}%
\typeout{** loaded for the language `#1'. Using the pattern for}%
\typeout{** the default language instead.}%
\else
\language=\csname l@#1\endcsname
\fi
#2}}
\providecommand{\BIBdecl}{\relax}
\BIBdecl

\bibitem{9296294}
Y.~Yaz{\i}c{\i}o{\u{g}}lu and A.~Speranzon, ``High dimensional robust consensus
  over networks with limited capacity,'' \emph{IEEE Control Systems Letters},
  vol.~5, no.~6, pp. 2024--2029, 2020.

\bibitem{review}
F.~Rossi, S.~Bandyopadhyay, M.~Wolf, and M.~Pavone, ``Review of multi-agent
  algorithms for collective behavior: A structural taxonomy,''
  \emph{IFAC-PapersOnLine}, vol.~51, no.~12, pp. 112--117, 2018.

\bibitem{ZlizDuan}
Z.~Li, Z.~Duan, G.~Chen, and L.~Huang, ``Consensus of multi-agent systems and
  synchronization of complex networks: A unified viewpoint,'' \emph{IEEE
  Transactions on Circuits and Systems}, vol.~57, no.~1, pp. 213--224, 2009.

\bibitem{Farhat_2021}
O.~Farhat, D.~Abou~Jaoude, and M.~Hudoba~de Badyn, ``$\mathcal{H}_\infty$
  network optimization for edge consensus,'' \emph{European Journal of
  Control}, vol.~62, pp. 2--13, 2021, 2021 European Control Conference Special
  Issue.

\bibitem{9115210}
D.~R. Foight, M.~Hudoba~de Badyn, and M.~Mesbahi, ``Performance and design of
  consensus on matrix-weighted and time-scaled graphs,'' \emph{IEEE
  Transactions on Control of Network Systems}, vol.~7, no.~4, pp. 1812--1822,
  2020.

\bibitem{CHENG20172451}
X.~Cheng and J.~Scherpen, ``Balanced truncation approach to linear network
  system model order reduction,'' \emph{IFAC-PapersOnLine}, vol.~50, no.~1, pp.
  2451--2456, 2017.

\bibitem{AbouJaoudeBalancedTruncation}
D.~{Abou Jaoude} and M.~Farhood, ``Balanced truncation model reduction of
  nonstationary systems interconnected over arbitrary graphs,''
  \emph{Automatica}, vol.~85, pp. 405--411, 2017.

\bibitem{AbouJaoudeLPV}
------, ``Model reduction of distributed nonstationary {LPV} systems,''
  \emph{European Journal of Control}, vol.~40, pp. 27--39, 2018.

\bibitem{AbouJaoudeCoprimeFactors}
------, ``Coprime factors model reduction of spatially distributed {LTV}
  systems over arbitrary graphs,'' \emph{IEEE Transactions on Automatic
  Control}, vol.~62, no.~10, pp. 5254--5261, 2017.

\bibitem{AbouJaoudeCoprimeFactorsIJC}
------, ``Coprime factors reduction of distributed nonstationary {LPV}
  systems,'' \emph{International Journal of Control}, vol.~92, no.~11, pp.
  2571--2583, 2019.

\bibitem{doi:10.1146/annurev-control-061820-083817}
X.~Cheng and J.~Scherpen, ``Model reduction methods for complex network
  systems,'' \emph{Annual Review of Control, Robotics, and Autonomous Systems},
  vol.~4, pp. 425--453, 2021.

\bibitem{8795611}
X.~Cheng, L.~Yu, and J.~Scherpen, ``Reduced order modeling of linear consensus
  networks using weight assignments,'' in \emph{18th European Control
  Conference ({ECC})}, 2019, pp. 2005--2010.

\bibitem{cheng2020reduced}
X.~Cheng, L.~Yu, D.~Ren, and J.~Scherpen, ``Reduced order modeling of
  diffusively coupled network systems: An optimal edge weighting approach,''
  \emph{arXiv preprint arXiv:2003.03559}, 2020.

\bibitem{S3}
N.~Monshizadeh, H.~Trentelman, and M.~K. Camlibel, ``Projection-based model
  reduction of multi-agent systems using graph partitions,'' \emph{IEEE
  Transactions on Control of Network Systems}, vol.~1, no.~2, pp. 145--154,
  2014.

\bibitem{convexconcave}
Q.~T. Dinh, S.~Gumussoy, W.~Michiels, and M.~Diehl, ``Combining convex-concave
  decompositions and linearization approaches for solving {BMI}s, with
  application to static output feedback,'' \emph{{IEEE} Transactions on
  Automatic Control}, vol.~57, no.~6, pp. 1377--1390, 2011.

\bibitem{BEFB:94}
S.~Boyd, L.~Ghaoul, E.~Feron, and V.~Balakrishnan, \emph{Linear Matrix
  Inequalities in System and Control Theory}.\hskip 1em plus 0.5em minus
  0.4em\relax Society for Industrial and Applied Mathematics, 1994.

\bibitem{boyd_vandenberghe_2004}
S.~Boyd and L.~Vandenberghe, \emph{Convex Optimization}.\hskip 1em plus 0.5em
  minus 0.4em\relax Cambridge University Press, 2004.

\bibitem{Vandeberghe}
G.~Pipeleers, B.~Demeulenaere, J.~Swevers, and L.~Vandenberghe, ``Extended
  {LMI} characterizations for stability and performance of linear systems,''
  \emph{Systems \& Control Letters}, vol.~58, no.~7, pp. 510--518, 2009.

\bibitem{Lofberg2004}
J.~Lofberg, ``{YALMIP}: A toolbox for modeling and optimization in {MATLAB},''
  \emph{IEEE International Conference on Robotics and Automation}, pp.
  284--289, 2004.

\bibitem{MOSEK}
\BIBentryALTinterwordspacing
M.~ApS, \emph{The MOSEK optimization toolbox for MATLAB manual. Version 9.3.6},
  2021. [Online]. Available:
  \url{https://docs.mosek.com/9.3/toolbox/index.html}
\BIBentrySTDinterwordspacing

\end{thebibliography}
\end{document}